\documentclass[runningheads,envcountsame]{llncs} % for arxiv

\usepackage{amsmath}
\spnewtheorem{observation}{Observation}{\bfseries}{\itshape}

\usepackage{amsthm}
\usepackage{amsfonts}
\newenvironment{claimproof}
  {\begin{proof}}
  {\end{proof}}

\usepackage[T1]{fontenc}
\usepackage{array}
\usepackage{tabularx}
\usepackage{graphicx}
\usepackage[font=small, labelfont=bf, labelsep=period]{caption}
\usepackage[font=small, labelfont=small]{subcaption}
\usepackage[dvipsnames]{xcolor}
\usepackage{tcolorbox}
\usepackage{xspace}
\usepackage{cite}
\usepackage{amsmath}
\usepackage[hidelinks]{hyperref}
\usepackage[capitalise,nameinlink,noabbrev]{cleveref}
\usepackage{thmtools}
\usepackage{orcidlink}
\usepackage{marvosym} % for \Letter (corresponding author)
\usepackage{enumitem}

\crefname{observation}{Observation}{Observations}
\crefname{algocf}{Algorithm}{Algorithms}

\usepackage{tikz}
\usetikzlibrary{arrows,math,patterns,shapes,fit,decorations.pathreplacing,shapes.geometric}

\definecolor{defblue}{rgb}{0.121,0.47,0.705}
\definecolor{linkblue}{rgb}{0.098,0.098,0.4392}
\definecolor{lipicsGray}{gray}{0.5}
\DeclareTextFontCommand{\defterm}{\color{defblue}\itshape}

\usepackage[algoruled,vlined,longend]{algorithm2e}
\SetKw{KwAnd}{and}
\SetKw{KwOr}{or}
\DontPrintSemicolon

\let\rho\varrho

\newcommand{\brac}[1]{{\left(#1\right)}}

\newcommand{\set}[1]{\left\{#1\right\}}

\newcommand{\norm}[1]{{\left|#1\right|}}

\newcommand{\ceil}[1]{{\left\lceil #1 \right\rceil}}

\newcommand{\Oh}[1]{O\brac{#1}}

\newcommand{\Real}{\mathbb{R}}

\newcommand{\I}{\ensuremath{\mathcal{I}}\xspace}

\newcommand{\none}{\mathrm{none}}
\newcommand{\rightend}{\mathrm{right}}
\newcommand{\leftend}{\mathrm{left}}

\newcommand{\tw}{\ensuremath{\mathsf{tw}}\xspace}

\newcommand{\defproblem}[4]{
  \begin{tcolorbox}[left=1.5mm,right=1.5mm]%
    \vspace*{-1ex}%
    \noindent
    \begin{tabularx}{\linewidth}{@{}>{\normalsize}l@{~~}>{\normalsize}X@{}}
      {\sf\bfseries\color{lipicsGray} Problem:} & #1\\[.1ex]
      {\sf\bfseries\color{lipicsGray} Input:} & #2\\[.1ex]
      {\sf\bfseries\color{lipicsGray} #4:} & #3
    \end{tabularx}\vspace*{-1ex}
  \end{tcolorbox}
}
\IfFileExists{latexml.sty}{\usepackage{latexml}}{\newif\iflatexml\latexmlfalse}
\iflatexml
\renewcommand{\defproblem}[4]{%
  \par\noindent
  \begingroup\setlength{\fboxsep}{0pt}% no padding, else a white ring appears between frame and background
  \fcolorbox{black!75}{black!5}{\colorbox{black!5}{% tcolorbox default colors; \colorbox for arXiv's LaTeXML, whose \fcolorbox ignores the background
  \begin{tabular}{l p{0.8\linewidth}}
    {\sf\bfseries\color{lipicsGray} Problem:} & #1\\
    {\sf\bfseries\color{lipicsGray} Input:} & #2\\
    {\sf\bfseries\color{lipicsGray} #4:} & #3
  \end{tabular}}}\endgroup\par
}
\fi
\iflatexml
\makeatletter
\newtheorem*{lxrestated}{\lxrestatedname}
\def\restatable{\@ifnextchar[\lx@rs@a{\lx@rs@a[]}}
\def\lx@rs@a[#1]#2#3{\def\lx@rs@args{{#1}{#2}{#3}}\lx@rs@b}
\long\def\lx@rs@b#1\end{\expandafter\lx@rs@c\lx@rs@args{#1}\end}
\def\lx@claim{claim}
\long\def\lx@rs@c#1#2#3#4{%
  \begin{#2}[#1]#4\end{#2}%
  \def\lx@env{#2}%
  \ifx\lx@env\lx@claim
    \expandafter\gdef\csname lx@#3@name\endcsname{\textit{\claimname}}%
  \else
    \expandafter\xdef\csname lx@#3@name\endcsname{\noexpand\textbf{\csname #2name\endcsname~\csname the#2\endcsname}}%
  \fi
  \expandafter\gdef\csname #3\endcsname{\@ifstar{\lx@restate{#3}{#1}{#4}}{}}}

\newcommand{\lx@restate}[3]{%
  \begingroup\let\label\@gobble
  \def\lxrestatedname{\csname lx@#1@name\endcsname}%
  \begin{lxrestated}[#2]#3\end{lxrestated}\endgroup}
\let\lx@llncs@keywords\keywords
\def\lx@kw@split#1\and#2\lx@kw@end{\g@addto@macro\lx@kw@list{#1}%
  \ifx\relax#2\relax\else\g@addto@macro\lx@kw@list{\unskip\ \textperiodcentered\ \ignorespaces}\lx@kw@split#2\lx@kw@end\fi}
\renewcommand{\keywords}[1]{\def\lx@kw@list{}\lx@kw@split#1\and\lx@kw@end
  \expandafter\lx@llncs@keywords\expandafter{\lx@kw@list}}
\makeatother
\fi
\newcommand{\defdecproblem}[3]{\defproblem{#1}{#2}{#3}{Question}}

\newcommand{\POL}{\ensuremath{\mathsf{P}}\xspace}
\newcommand{\NP}{\ensuremath{\mathsf{NP}}\xspace}

\newcommand{\PHS}[1]{\ensuremath{\Sigma^\mathsf{P}_{#1}}\xspace}

\newcommand{\SPTWO}{\PHS{2}}

\newcommand{\PSPACE}{\ensuremath{\mathsf{PSPACE}}\xspace}

\newcommand{\EXPTIME}{\ensuremath{\mathsf{EXPTIME}}\xspace}

\newcommand{\PSPTWOSAT}{\textsc{Existential-}2\textsc{-Level-SAT}\xspace}

\DeclareMathOperator{\payoff}{\mathcal{U}_\mathrm{def}}
\DeclareMathOperator{\payoffatt}{\mathcal{U}_\mathrm{att}}
\newcommand{\Ppuredefpureatt}{\textsc{CP(Pure Defense, Any Pure Attack)}\xspace}
\newcommand{\Ppuredeffixatt}{\textsc{CP(Fixed Attack, Pure Defense)}\xspace}
\newcommand{\Ppuredefmixatt}{\textsc{CP(Fixed Mixed Attack, Pure Defense)}\xspace}

\newcommand{\Ppureattpuredef}{\textsc{CP(Pure Attack, Any Pure Defense)}\xspace}
\newcommand{\Ppureattfixdef}{\textsc{CP(Fixed Defense, Pure Attack)}\xspace}
\newcommand{\Ppureattmixdef}{\textsc{CP(Fixed Mixed Defense, Pure Attack)}\xspace}

\newcommand{\Pmixdefmixatt}{\textsc{CP(Mixed Defense, Mixed Attack)}\xspace}

\newcommand{\Pclique}{\textsc{Clique}\xspace}
\newcommand{\Psetcover}{\textsc{SetCover}\xspace}
\newcommand{\Pcliquenodedel}{\textsc{CliqueNodeDeletion}\xspace}
\newcommand{\Pbalancedvertexseparator}{\textsc{BalancedVertexSeparator}\xspace}

\newif\ifinappendix% Default is \inappendixfalse
\let\oldappendix\appendix% Store \appendix
\renewcommand{\appendix}{% Update \appendix
  \oldappendix% Default \appendix
  \renewcommand{\thesection}{\Alph{section}}% As in llncs; LaTeXML (arXiv HTML) would give 0.A
  \renewcommand{\theHsection}{appendix.\Alph{section}}% Unique hyperref anchors (else clash with Sections 1-3)
  \inappendixtrue% Set switch to true
}

\graphicspath{{figures/}}
\newcommand{\arxiv}[1]{#1}
\newcommand{\lncs}[1]{}

\arxiv{\hypersetup{
  pdftitle={Modelling Network Resilience: The Complexity of Some Graph Division Games},
  pdfauthor={Grzegorz Gutowski, Konstanty Junosza-Szaniawski, Antonio Lauerbach, Alexander Wolff},
  pdfkeywords={controller placement, fair division, cake cutting, attacker-defender game},
  bookmarksdepth=subsection % llncs sets tocdepth=0, which hides sections in the PDF outline
}}

\providecommand{\Description}[1]{}

\arxiv{\newcommand{\restateref}[1]{\ifinappendix{\hyperref[#1]{$\star$}}\else{\hyperref[#1*]{$\star$}}\fi}}
\lncs{\newcommand{\restateref}[1]{\href{https://arxiv.org/pdf/2605.17572}{$\star$}}}

\begin{document}

\title{Modelling Network Resilience:\texorpdfstring{\\}{ }
  The Complexity of Some Graph Division Games}

\titlerunning{The Complexity of Some Graph Division Games}

\author{Grzegorz Gutowski\inst{1}\lncs{\textsuperscript{(\Letter)}\orcidID{0000-0003-3313-1237}}\arxiv{\orcidlink{0000-0003-3313-1237}} \and
  Konstanty Junosza-Szaniawski\inst{2}\lncs{\orcidID{0000-0003-0352-8583}}\arxiv{\orcidlink{0000-0003-0352-8583}}
  \and\\
  Antonio~Lauerbach\inst{3}\lncs{\orcidID{0009-0007-9093-3443}}\arxiv{\orcidlink{0009-0007-9093-3443}} \and
  Alexander~Wolff\inst{3}\lncs{\orcidID{0000-0001-5872-718X}}\arxiv{\orcidlink{0000-0001-5872-718X}}}

\institute{
  Institute~of~Theoretical~Computer~Science, Faculty of Mathematics and Computer~Science, Jagiellonian~University, Krak{\'o}w, Poland\lncs{\\\email{grzegorz.gutowski@uj.edu.pl}} \and
  Warsaw University of Technology, Warsaw, Poland \and
  Universit\"at W\"urzburg, W\"urzburg, Germany
}

\lncs{\authorrunning{G.~Gutowski et al.}}
\arxiv{\authorrunning{G.~Gutowski, K.~Junosza-Szaniawski, A.~Lauerbach, and A.~Wolff}}

\maketitle

\begin{abstract}
Motivated by the controller placement problems in software-defined networks and the fair division principles of classical ``cake cutting'', we investigate the following two-player zero-sum game.
In our model, a defender places a limited number of controllers on graph vertices, while an attacker deletes a limited number of vertices.
The defender score is the total number of surviving vertices reachable from any remaining controller.
We formalize the computational problems associated with various game dynamics (defender plays first; attacker plays first; players play simultaneously; pure or mixed strategies).\looseness=-1

We show that these natural problems are \NP-complete or \SPTWO-complete, depending on the specific variant.
These hardness results impose limitations on optimal controller placement algorithms under different notions of quality of a solution.
Finally, we present structural insights that yield efficient
algorithms for restricted graph classes (namely interval graphs and
graphs of bounded treewidth).
\end{abstract}

\section{Introduction}

Cake cutting is a classical problem in fair division~\cite{Steinhaus1948,DubinsSpanier1961,BramsTaylor1996,RobertsonWebb1998,Procaccia2016}. Its goal is to divide a heterogeneous divisible resource, called a \defterm{cake}, among several agents with different preferences. In the standard model, the cake is represented by the interval $[0,1]$, and each agent has a valuation function over pieces of the cake. For two agents, a well-known procedure is the \defterm{I~cut, you choose} rule: one agent cuts the cake into two pieces, and the other chooses one of them. For three or more agents, however, the problem becomes significantly more difficult. In particular, Stromquist~\cite{Stromquist} showed that no finite protocol can guarantee an envy-free division when each agent is required to receive a single connected piece of the interval.

A natural extension of the model is to replace the interval by a graph. In this setting, the resource is distributed over the graph, and each agent is typically required to receive a connected piece. This captures situations where the resource has an underlying network structure and disconnected allocations are not allowed. Unlike the interval case, the graph setting introduces additional difficulties because the topology of the graph limits the possible allocations. Therefore, fairness and computational complexity both become closely tied to the structure of the graph and the agents' valuations~\cite{BouveretEtAl2017,BeiEtAl2022}.

This perspective becomes even more natural when one views division as
an interactive process rather than a static outcome. Already in
classical cake cutting, procedures such as \emph{I~cut, you choose}
involve sequential decisions made by the agents.
On graphs, similar ideas lead to strategic settings in which players
with conflicting objectives can influence the resulting division. The
problem then is no longer only to determine whether a fair or good
division exists, but also to understand which player can enforce a
desirable outcome.

Once division on graphs is viewed as a strategic process, it becomes natural to study its algorithmic complexity from a game-theoretic perspective. The main question is whether one of the players can guarantee a desired outcome against all possible responses of the opponent. Such problems are often computationally hard, because evaluating a move requires reasoning about the entire game tree. Classical examples include the \PSPACE-completeness of \emph{Generalized Geography}~\cite{Schaefer1978} and \emph{Hex}~\cite{Reisch1981}, and the \EXPTIME-completeness of generalized \emph{Chess} and \emph{Checkers}~\cite{FraenkelLichtenstein1981,Robson1984}.

In this paper, we study a game introduced by Fortz, Mycek, Pi{\'o}ro, and Toma\-szews\-ki~\cite{TomaszewskiPM2022,fortz}, motivated by the problem of placing controllers in a telecommunication network so as to improve its resilience to attacks. A network operator places controllers at selected vertices of the network, while an attacker chooses a set of vertices to remove and disconnect the network. A vertex is said to survive if it belongs to a connected component containing at least one controller. The outcome of the game is the number of surviving vertices. This game can be viewed as a variant of the classical \emph{I cut, you choose} procedure, with the crucial difference that the chooser must act before knowing where the cuts will occur.

From a game-theoretic perspective, a \defterm{pure strategy} of the
defender is a choice of controller locations, while a pure strategy of
the attacker is a choice of attacked vertices.  The payoff is
determined by the number of surviving vertices for the pair of
strategies chosen by the players.  A \defterm{mixed strategy} is a
probability distribution over pure strategies, meaning that a player
randomizes between several possible choices.  Thus, if the complete
strategy sets of both players are given explicitly as part of the
input, then the game can be represented by a payoff matrix and its
value can be computed in polynomial time.  In general, however, the
number of possible strategies of each player may be exponential in the
size of the graph.  To cope with this difficulty, one may apply the
double-oracle method.  Starting from restricted sets of strategies for
both players, one repeatedly computes a best response to the current
strategy set of the opponent and adds this best response to the
corresponding set.  The process continues until no player can improve
by adding a new strategy.  In practice, best-response problems are
often solved using mixed-integer linear programming, and the overall
approach is closely related to the column-generation
method~\cite{McMahan2003536}.  The game was originally studied as a
min--max game with pure strategies, while Junosza-Szaniawski and
Nogalski~\cite{JunoszaN2023} also considered mixed strategies for
restricted strategy sets.  This approach was extended by applying the
double-oracle method to the mixed-strategy version of the
game~\cite{PioroMTJN2023}.

The double-oracle method is frequently used in robust optimization, where one seeks solutions that perform well under uncertainty in the input data~\cite{ben2009robust,bertsimas2011theory}. In this sense, the game studied here can be interpreted as an optimal placement problem under uncertainty regarding the set of vertices attacked by the adversary.

\paragraph{Our Contribution}

The main contribution of this paper is a complexity classification of
several natural variants of the controller-placement game on graphs.
We show that the defender's optimization problem against a fixed
attack is solvable in linear time (\cref{obs:optdef}), whereas the
attacker's optimization problem against a fixed defense is \NP-complete
(\cref{thm:pureattfixdef}).  For the sequential game, we prove that
deciding whether the defender has a strategy that guarantees a
prescribed value is \SPTWO-complete (\cref{thm:puredefpureatt}), whereas
the analogous problem for the attacker is \NP-complete
(\cref{thm:pureattpuredef}).  We further study mixed-strategy settings
arising from double-oracle approaches and show that the associated
best-response problems are \NP-complete; see \cref{sec:mixed}.  Finally,
we provide structural results that yield efficient
algorithms on interval graphs and graphs of bounded treewidth; see
\cref{sec:restricted}.
The full proofs of statements marked with a (clickable) star
($\star$) can be found in the \arxiv{appendix}\lncs{full
  version~\cite{gjlw-mnr-arXiv26}}.

\paragraph{Related Work}

Fair division of indivisible goods under connectivity constraints has
been studied intensively in recent years;
surveys~\cite{Suksompong2021,BiswasEtAl2023} are available.  In this
setting, goods are represented as vertices of a graph and each agent
must receive a connected bundle.
(Somewhat related is the \defterm{critical node problem}, where the
aim is to remove a given number of nodes so as to minimize the sum of
the squared sizes of the resulting connected
components~\cite{Arulselvan-COR09}.)  Still in this setting, the
\defterm{maximin share} has been studied~\cite{BouveretEtAl2017}, that
is, the value an agent can guarantee by acting as the cutter in an
\emph{I cut, you choose} procedure.  It was
shown~\cite{BouveretEtAl2017} that maximin share allocations always
exist on trees, but may fail to exist already on cycles.  Subsequent
work established approximation guarantees for several graph classes,
including cycles~\cite{LoncTruszczynski2020},
$d$\nobreakdash-claw-free graphs~\cite{Lonc2023,Lonc26}, block graphs,
cacti, complete multipartite graphs, and split graphs \cite{Lonc25}.
Others introduced structural notions such as the price of
connectivity~\cite{BeiEtAl2022}.  Beyond maximin share, the literature
also considers other fairness and efficiency notions, including
envy-freeness and its relaxations~\cite{BeiEtAl2022,BiloEtAl2022},
chore division~\cite{BouveretCechlarovaLesca2019,XiaoEtAl2023},
Pareto-optimal connected allocations~\cite{IgarashiPeters2019}, local
fairness~\cite{HummelIgarashi2024}, and computational complexity
questions for graph-based fair
division~\cite{GrecoScarcello2020,DeligkasEtAl2021}.

Many two-player graph games can be expressed as existential--universal
decision problems: one asks whether there exists a move or strategy
for one player such that every response of the opponent still leads to
the desired outcome.  This alternating-quantifier structure is a
common source of $\PHS{2}$-hardness (see \cref{sec:preliminaries}
and \arxiv{Appendix~\ref{app:sptwo}}\lncs{the full
  version~\cite{gjlw-mnr-arXiv26}}).  Several related graph problems have
this form.  Bliem and Woltran~\cite{BliemW18} showed that
\textsc{Secure Set} and several of its variants are
$\PHS{2}$-complete.  Cardinal and Joret~\cite{Cardinal} proved that
deciding whether a bipartite graph has a set of at most $k$ vertices
hitting all maximal independent sets is $\PHS{2}$-complete.
Marx~\cite{MARX} established that $k$-\textsc{Clique-Coloring} is
$\PHS{2}$-complete for every $k \ge 2$.  A particularly close
example is the recent work of Chaplick, Gutowski, and
Krawczyk~\cite{ChaplickEtAl25NoteComplexityDefensiveDomination}, who
proved $\PHS{2}$-completeness for \textsc{DefensiveDominatingSet}.
More recently, Gr\"une and Wulf~\cite{GRUNE} developed a general
framework for many \defterm{interdiction problems}, where the aim if
to remove as few vertices or edges from a graph so that all structures
of a given size are destroyed.  They showed $\PHS{2}$-completeness for
many interdiction problems, including clique, independent set, and
Hamiltonian path/cycle interdiction.

Our game can also be viewed through an existential--universal lens, but it is strongly motivated by controller placement problems in software-defined networks. This literature considers a broad range of objectives, including communication delay~\cite{DouQYG23,WangZHW18}, controller capacity~\cite{IbrahimHNSNF21,YaoBLG14}, dynamic traffic and topology changes~\cite{BariCZZAB13,ToufgaAAOV20,HeVK19,DixitHMLK13}, and resilience to failures and attacks~\cite{WangChen21,tohidi,kumar,Sebopelo_Isong_2024,SantosGT21}; see~\cite{DasSG2020,SinghS2018survey} for surveys. 

A similar existential--universal structure appears in robust optimization: one asks whether there exists a solution that remains good for all realizations of the uncertain data from a given uncertainty set. Robust optimization is a central framework for optimization under uncertainty~\cite{ben2009robust,bertsimas2011theory}, and in combinatorial network problems it often takes an adversarial form, where the uncertainty is interpreted as an opponent selecting the most damaging scenario~\cite{BertsimasS03}. Because this combines combinatorial structure with worst-case analysis over many possible scenarios, the resulting models are usually significantly harder than their nominal versions and often require iterative solution methods such as double oracle.

\section{Preliminaries}
\label{sec:preliminaries}

For a graph~$G$, let $V(G)$ denote the set of vertices of~$G$ and let
$E(G) \subseteq \binom{V(G)}{2}$ denote the set of edges of~$G$.
In a graph, a \defterm{($k$-)defense} is a set of (at most~$k$) vertices,
and an \defterm{($\ell$-)attack} is a set of (at most~$\ell$) vertices.
Let $\mathcal{D}_k(G)=\{D\subseteq V(G): |D|\leq k\}$ and $\mathcal{A}_\ell(G)=\{A\subseteq V(G): |A|\leq\ell\}$ denote the sets of $k$-defenses and $\ell$-attacks.
Given a graph~$G$ with a defense~$D$ and an attack~$A$, we say that a
vertex $v \in V(G) \setminus A$ \defterm{survives} if the graph $G - A$
contains a path from $v$ to some vertex in~$D$.
If a vertex does not survive under defense~$D$ and attack~$A$, we say
that it gets \defterm{disabled}.
We define the \defterm{payoff function} $\payoff(G,D,A)$ for
the defender to be the number of vertices that survive under
defense~$D$ and attack~$A$.
Correspondingly, we define the \defterm{payoff function} $\payoffatt(G,D,A)$ for the attacker to be the number of
vertices that get disabled.
Clearly, we have $\payoff(G,D,A) + \payoffatt(G,D,A) = \norm{V(G)}$,
so the game is \defterm{zero-sum}.  The defender seeks to
maximize~$\payoff$, while the attacker seeks to maximize~$\payoffatt$
in order to minimize~$\payoff$.  Moreover, every vertex in~$A$ is disabled
by the attack, hence $\payoffatt(G,D,A) \ge \norm{A}$.
Note that attacking additional vertices can only increase~$\payoffatt$, and defending additional vertices can only increase~$\payoff$. Furthermore, if $k>\norm{V(G)}$ the defender can just defend all vertices, similarly, the attacker can attack all vertices if $\ell>\norm{V(G)}$. Hence, we can assume w.l.o.g.\ that $k\leq \norm{V(G)}$ and $\ell\leq \norm{V(G)}$ and that $k=\norm{D}$ and $\ell=\norm{A}$ for considered $k$-defenses $D$ and $\ell$-attacks $A$.

A \defterm{mixed $k$-defense} $\mathcal{D}$ is a multiset
$\set{(D,p_D) : D\in\mathcal{D}_k(G)}$, where each~$p_D$ is a
non-negative rational number and $\sum_{D\in\mathcal{D}_k(G)} p_D = 1$.
A \defterm{mixed $\ell$-attack} $\mathcal{A}$ is a multiset
$\set{(A,q_A): A\in \mathcal{A}_\ell(G)}$, where each~$q_A$ is a
non-negative rational number and $\sum_{A\in \mathcal{A}_\ell(G)} q_A = 1$.
Let $\mathcal{D}_k^{\textrm{mix}}(G)$ and $\mathcal{A}_\ell^{\textrm{mix}}(G)$ denote the set of all mixed $k$-defenses and all mixed $\ell$-attacks, respectively. 
For $\mathcal{D}\in \mathcal{D}_k^{\textrm{mix}}(G)$ and $\mathcal{A}\in \mathcal{A}_\ell^{\textrm{mix}}(G)$, we define the \defterm{expected payoff} 
\[\payoff(G,\mathcal{D},\mathcal{A})=\sum_{D\in \mathcal{D}_k(G)}\sum_{A\in \mathcal{A}_\ell(G)} p_D\cdot q_A\cdot\payoff(G,D,A).\]
The \defterm{value of the game} is defined by
\begin{align*}
  \operatorname{val}(G,k,\ell)
  =& \max_{\mathcal{D}\in \mathcal{D}_k^{\mathrm{mix}}(G)}
  \min_{\mathcal{A}\in \mathcal{A}_\ell^{\mathrm{mix}}(G)}
  \payoff(G,\mathcal{D},\mathcal{A}) \\
  =& \min_{\mathcal{A}\in \mathcal{A}_\ell^{\mathrm{mix}}(G)}
  \max_{\mathcal{D}\in \mathcal{D}_k^{\mathrm{mix}}(G)}
  \payoff(G,\mathcal{D},\mathcal{A}).
\end{align*}
The second equality follows from the minimax theorem for finite two-player zero-sum games~\cite{vonNeumann1928}.
A pair $(\mathcal{D}^*,\mathcal{A}^*)\in \mathcal{D}_k^{\mathrm{mix}}(G) \times \mathcal{A}_\ell^{\mathrm{mix}}(G)$
is a \defterm{Nash equilibrium} if
\begin{align*}
\forall \mathcal{D}\in \mathcal{D}_k^{\mathrm{mix}}(G) \colon
& \forall \mathcal{A}\in \mathcal{A}_\ell^{\mathrm{mix}}(G) \colon \\
& \payoff(G,\mathcal{D},\mathcal{A}^*)
\le
\payoff(G,\mathcal{D}^*,\mathcal{A}^*)
\le
\payoff(G,\mathcal{D}^*,\mathcal{A})\text{.}
\end{align*}
In other words, neither player can improve by a unilateral change of the strategy.

\paragraph{Complexity}

Our objective is to understand the computational complexity of computing optimal strategies for the defender and the attacker under different game dynamics.
The complexity classes \NP and \SPTWO are the most relevant for our purposes.
In order to show hardness results, we use polynomial-time many-to-one reductions to show \NP-completeness and \SPTWO-completeness of natural decision problems associated with the optimization problems of interest.
For the class \SPTWO, we\arxiv{ give a very brief introduction in Appendix~\ref{app:sptwo} and} refer the reader to the textbook by Arora and Barak~\cite[Chapter~5]{AroraB09}.
As our problems are naturally expressed as two-round games and ask if there exists a move for the first player such that, for every move of the second player, a certain payoff is guaranteed, \SPTWO is a natural complexity class for some of our problems.

\paragraph{Problem Statements}

We first consider division games in which the players play one after another, and the second player can observe the move of the first player before making their own move.
In this case, we can distinguish two variants: one in which the defender plays first and one in which the attacker plays first.
We are interested in computing an optimal response of the second player to a given move of the first player.

\defdecproblem
{\Ppuredeffixatt}
{A graph $G$, an attack $A$, positive integers $k$ and $\delta$.}
{Is there a $k$-defense~$D$ such that $\payoff(G,D,A) \ge \delta$?}

\begin{restatable}[\restateref{obs:optdef}]{observation}{obsOptDef}
  \label{obs:optdef}
  \Ppuredeffixatt can be solved in time linear in the size of~$G$.
\end{restatable}

Symmetrically, we are interested in computing an optimal response to a
given defense, which corresponds to the following problem.

\defdecproblem
{\Ppureattfixdef}
{A graph $G$, a defense $D$, positive integers $\ell$ and $\alpha$.}
{Is there an $\ell$-attack~$A$ such that $\payoffatt(G,D,A) \ge \alpha$?}
We show that \Ppureattfixdef is \NP-complete; see \cref{thm:pureattfixdef} in \cref{sec:pure}.
It is an interesting phenomenon that this problem is computationally harder than \Ppuredeffixatt.

We are also interested in deciding whether there exists a
\emph{strategy} for the first player that guarantees a certain payoff
against \emph{any} response of the opponent.  As a consequence of the
above disparity, we can also expect that the complexities of deciding
optimal strategies for the first player differ.

\defdecproblem{\Ppuredefpureatt}
{A graph $G$, positive integers $k$, $\ell$, and $\delta$.}
{Is there a $k$-defense~$D$ such that, for every $\ell$-attack~$A$, $\payoff(G,D,A) \ge \delta$?}
We show that \Ppuredefpureatt is \SPTWO-complete (see \cref{thm:puredefpureatt}).

\defdecproblem
{\Ppureattpuredef}
{A graph $G$, positive integers $k$, $\ell$, and $\alpha$.}
{Is there an $\ell$-attack~$A$ such that, for every $k$-defense~$D$, $\payoffatt(G,D,A) \ge \alpha$?}
Due to \cref{obs:optdef}, it is easy to construct an optimal defense for a fixed attack.
Hence we cannot expect the same complexity for the variant of the game where
the attacker plays first.
Still, \Ppureattpuredef is \NP-complete;
see \cref{thm:pureattpuredef}.

When we switch our attention to randomized strategies (see
\cref{sec:mixed}), the most natural question is to find the value of
the game.  The corresponding decision problem is as follows.
\defdecproblem
{\Pmixdefmixatt}
{A graph $G$, positive integers $k$, $\ell$, and $\delta$.}
{Is there a mixed $k$-defense~$\mathcal{D}$ such that, for every mixed $\ell$-attack~$\mathcal{A}$, we have $\payoff(G,\mathcal{D},\mathcal{A}) \ge \delta$?}

\section{Pure Strategies}
\label{sec:pure}

In this section, we settle the complexity of the deterministic
problems that we defined in \cref{sec:preliminaries}.  We start with a
given defense: Can an $\ell$-attack secure a payoff of~$\alpha$?

\begin{theorem}\label{thm:pureattfixdef}
    \Ppureattfixdef is \NP-complete.
\end{theorem}

\begin{proof}
  Membership in \NP is obvious as we can compute $\payoffatt(G,D,A)$ in polynomial time.
  To show \NP-hardness, we reduce from \Pclique, 
  defined as follows.

\defdecproblem
{\Pclique}
{A graph $G$ and a positive integer $t$.}
{Does $G$ contain a $t$-clique, that is, a subgraph isomorphic to~$K_t$?}

Given an instance $\langle G,t \rangle$ of \Pclique with~$t\geq 2$, we set
$n=|V(G)|$, $m=|E(G)|$, $\alpha=t+\binom{t}{2}$, $\ell=t$, and
construct a graph~$H$ and a defense $D \subseteq V(H)$ such that there
exists an $\ell$-attack~$A$ with $\payoffatt(H,D,A) \ge \alpha$ if and
only if $G$ contains a $t$-clique.

Let $H$ be the graph that has a \defterm{node vertex}~$v'$ for every
node~$v$ of~$G$ and an \defterm{edge vertex}~$e'$ for every edge~$e$
of~$G$.  In~$H$, every two node vertices are adjacent, and for every
edge $e=\{u,v\}$ of~$G$, the edge vertex $e'$ is adjacent to the two
corresponding node vertices~$u'$ and~$v'$.
Let $D=\set{v' : v \in V(G)}$ be the set of node vertices of~$H$.

We say that a $t$-attack~$A$ is \defterm{successful} if $\payoffatt(H,D,A) \ge t+\binom{t}{2}$.
For $A$ to be successful, it has to disable at least
$\binom{t}{2}$ vertices apart from those in~$A$.
We show that there exists a successful $t$-attack if and
only if $G$ contains a $t$-clique.

If $G$ contains a $t$-clique $v_1, v_2, \dots, v_t$, then the attack
$A=\set{v'_1,v'_2,\dots,v'_t}$ disables, apart from the $t$ vertices
in~$A$, the $\binom{t}{2}$ vertices in the set
$\{e' : e = \set{v_i,v_j}, 1\le i < j \le t\}$.  Hence, $A$ is
successful.

For the other direction, let
$A=\set{v'_1,v'_2,\dots,v'_p,e'_1,e'_2,\ldots,e'_r}$ be a $t$-attack
with $t=p+r$.  Then $A$ disables the vertices in~$A$ and every edge
vertex~$e'$ (in $H$) such that both endpoints of~$e$ (in $G$) are in
$\set{v_1,v_2,\dots,v_p}$.  Hence, $\payoffatt(H,D,A) \le
 t+\binom{p}{2}$.  If~$A$ is successful, then
$\payoffatt(H,D,A) = t+\binom{t}{2}$.  Hence, in this case, it holds
that $p=t$ and $r=0$.  Also, each of the $\binom{t}{2}$ pairs of
vertices in $\set{v_1,v_2,\dots,v_t}$ must be connected by an edge
in~$G$ (so that $A$ disables the corresponding edge vertex).  This
means that these vertices induce a $t$-clique in~$G$.
\end{proof}

Next, we ask: Is there a $k$-defense with a payoff of~$\delta$ against
every possible $\ell$-attack?

\begin{theorem}\label{thm:puredefpureatt}
  \Ppuredefpureatt is \SPTWO-complete.
\end{theorem}

\begin{proof}
  Membership in \SPTWO is clear: Given a graph~$G$, a $k$-defense~$D$,
  and integers~$\ell$ and~$\delta$, we can use an \NP-oracle to test
  whether, for every $\ell$-attack~$A$, we have
  $\payoff(G,D,A) \ge \delta$.
  
  To show hardness, we reduce from \Pcliquenodedel, which is
  \SPTWO-hard~\cite{Rutenburg86ComplexityGeneralizedGraphColoring,
    ChaplickEtAl25NoteComplexityDefensiveDomination}.

  \defdecproblem
  {\Pcliquenodedel}
  {A graph $G$, positive integers $s$ and $t$.}
  {Is there a set~$X$ of at most $s$~vertices such that $G\setminus X$
    does not contain a $t$-clique?}

  Let~$\langle G,s,t \rangle$ be an instance of \Pcliquenodedel, let
  $n=|V(G)|$, and let $m=|E(G)|$.  Without loss of generality, we can assume
  that~$n \ge s > \binom{t}{2}+2t$ (otherwise we can obtain an equivalent instance
  with a larger~$s$ by adding 
  vertices
  adjacent to all vertices, thereby necessitating their removal).

  To obtain an instance of \Ppuredefpureatt, we set $k=s$, $\ell=t$,
  $\delta=2n+m-\big(\binom{t}{2}+2t\big)+1$, and construct a graph~$H$
  similarly as in the proof of \cref{thm:pureattfixdef}.  Again, $H$
  has a \defterm{node vertex}~$v'$ for every node~$v$ of~$G$ and an
  \defterm{edge vertex}~$e'$ for every edge~$e$ of~$G$.  In~$H$, every
  two node vertices are adjacent, and every edge vertex is adjacent to
  the two corresponding node vertices.  Additionally, for every
  node~$v$ of~$G$, $H$ has a \defterm{copy vertex}~$v''$ that is adjacent
  only to the corresponding node vertex~$v'$.
  \cref{fig:clique-node-deletion-reduction} depicts an example.

  \begin{figure}[tbp]
    \Description{Example instance (a path on three vertices) for the
      reduction from the problem CliqueNodeDeletion to the problem
      CP(Pure Defense, Any Pure Attack)}
    \hfil
    \begin{minipage}[b]{.28\textwidth}
      \centering
      \includegraphics[page=1]{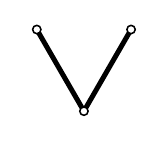}
      \subcaption{}
      \label{fig:original}
    \end{minipage}
    \hfil
    \begin{minipage}[b]{.47\textwidth}
      \centering
      \includegraphics[page=2]{clique-node-deletion-reduction}
      \subcaption{}
      \label{fig:reduced}
    \end{minipage}
    \hfil    
    \caption{Reduction from an instance of \Pcliquenodedel (a) to the corresponding instance of \Ppuredefpureatt (b): edges are subdivided by edge
      vertices (squares).  Node vertices (disks) are connected into a
      clique, and every node vertex has a private copy vertex
      (triangle).}
    \label{fig:clique-node-deletion-reduction}
  \end{figure}

  The resulting graph~$H$ contains~$2n+m$ vertices.
  Thus, given a $k$-defense~$D$ and an $\ell$-attack~$A$, it holds that $\payoff(H,D,A)\geq \delta$ if and only if~$\payoffatt(H,D,A)<\binom{t}{2}+2t$.
  Our aim is to show the following equivalence:  There is a set~$X$ of $s$ vertices in~$G$ such that~$G\setminus X$ contains no $t$-clique if and only if there exists an $s$-defense~$D$ in~$H$ such that $\payoffatt(H,D,A)<\binom{t}{2}+2t$ for every $t$-attack~$A$.
  The idea is that the defense~$D$ corresponds to~$X$, whereas an attack~$A$ with~$\payoffatt(H,D,A) \ge \binom{t}{2}+2t$ would correspond to a~$t$-clique in~$G\setminus X$.

  \begin{restatable}[\restateref{claim:node-vertex-survival}]{claim}{ClaimNodeVertexSurvival}
    \phantomsection% claims have no counter; gives the labels here and in the appendix an anchor
    \label{claim:node-vertex-survival}
    For every $s$-defense~$D$ and every $t$-attack~$A$, every node vertex that is not attacked directly survives.
  \end{restatable}

  It follows from this claim that a node vertex can only be disabled by being attacked directly. Thus, 
  attacking an edge vertex or a copy vertex yields only one disabled vertex.
  As attacking a node vertex disables at least one vertex,
  we can assume without loss of generality that any attack attacks only node vertices.

  Since we can assume that exclusively node vertices are attacked, a $t$-attack~$A$ can ensure $\payoffatt(H,D,A) \ge \binom{t}{2}+2t$ if and only if (i)~the~$t$ vertices in~$G$ that correspond to the node vertices in~$A$ form a clique and (ii)~none of the edge vertices of that clique or the copies of node vertices in~$A$ are defended.
  Thus, if~$\langle G,s,t \rangle$ is a yes-instance of \Pcliquenodedel with deletion set~$X$, we let $D=\set{v''\mid v\in X}$ be our defense.
  Suppose that there is an attack~$A$ that disables at least~$\binom{t}{2}+2t$ vertices.
  As we have argued, this implies that the vertices attacked by~$A$ form a $t$-clique in~$G$, and that none of their copies are defended.
  Thus, they form a $t$-clique in~$G\setminus X$, a contradiction to~$X$ being a solution to the \Pcliquenodedel instance.
  Therefore, for any attack~$A$, it holds that~$\payoff(H,D,A)\geq \delta$.

  For the other direction, let~$D$ be a defense that is
  \defterm{successful}, that is, for any attack~$A$,
  $\payoff(H,D,A) \ge \delta$.
  We now show that we can assume, without loss of generality, that~$D$ defends only copy vertices.
  First, suppose that~$D$ defends a node vertex~$v'$.
  Then, we can instead defend its copy~$v''$ without increasing the number of disabled vertices under any attack~$A$ (recall that $v''$ is never attacked directly, as we assume that only node vertices are attacked).
  Indeed, if~$v'$ is attacked in~$A$, this change leads to one additional vertex, the copy~$v''$, surviving.
  Otherwise, if~$v'$ is not attacked, then~$v'$ and~$v''$ survive under both defenses. 
  Second, suppose that~$D$ defends an edge vertex~$e'$, where~$e=\{v,w\}$.
  We claim that $D'=(D \setminus \{e'\}) \cup \{v''\}$ is also successful.
  Indeed, suppose that there is an attack~$A$ under~$D'$ that disables at least~$\binom{t}{2}+2t$ vertices.
  As we have argued before, this implies that the $t$ vertices in~$A$ correspond to a $t$-clique~$K$ in~$G$, and that none of the copy or edge vertices corresponding to~$K$ are defended.
  Therefore,~$v$, and by extension~$e$, cannot be part of~$K$.
  As the defense is unchanged except for~$e'$ and~$v''$, this implies that~$A$ also disables at least~$\binom{t}{2}+2t$ vertices under~$D$, a contradiction.
  Therefore, we can assume that~$D$ defends only copy vertices.
  Given a successful defense~$D$ consisting exclusively of copy vertices, let~$X=\set{v \mid v'' \in D}$.
  By construction, the graph~$G\setminus X$ contains no $t$-clique~-- otherwise attacking a $t$-clique would disable~$\binom{t}{2}+2t$ vertices, a contradiction.
\end{proof}

Finally, we ask: Is there an $\ell$-attack that achieves a payoff
of~$\alpha$ against every $k$-defense?

\begin{theorem}
  \label{thm:pureattpuredef}
  \Ppureattpuredef is \NP-complete.
\end{theorem}
\begin{proof}
  Membership in \NP is due to \cref{obs:optdef}.
  To show \NP-hardness, we reduce from the following \NP-hard
  problem~\cite{BuiJ92}.
  \defdecproblem
  {\Pbalancedvertexseparator}
  {A graph $G$ and a positive integer $h$.}
  {Is there a set~$X$ of at most $h$ vertices such that every connected
    component of~$G\setminus X$ has size at most $\lceil\frac{n-h}{2}\rceil$,
    where $n=|V(G)|$?}

  Given an instance $\langle G,h \rangle$ of \Pbalancedvertexseparator, we set $n=|V(G)|$, $k=1$, $\ell=h$, $\delta=\ceil{\frac{n-h}{2}}$, and $\alpha=n-\delta$, and construct an instance $\langle G,k,\ell,\alpha \rangle$ of \Ppureattpuredef.

  Let $\langle G,h \rangle$ be a yes-instance of \Pbalancedvertexseparator. After attacking the vertices in~$X$, every connected component of the remaining graph has size at most $\lceil\frac{n-h}{2}\rceil$. Hence, for any defense~$D$ of size~$1$, we have $\payoff(G,D,X) \le \ceil{\frac{n-h}{2}} = \delta = n-\alpha$. Thus, $\langle G,k,\ell,\alpha \rangle$ is a yes-instance of \Ppureattpuredef.

  When $\langle G,k,\ell,\alpha \rangle$ is a yes-instance of
  \Ppureattpuredef, there is an attack~$A$ of size~$h$ such that for
  any defense~$D$ of size~$1$, we have
  $\payoff(G,D,A) \le \delta = \ceil{\frac{n-h}{2}}$.  In particular,
  this implies that the largest connected component of~$G\setminus A$
  has size at most $\ceil{\frac{n-h}{2}}$.  Hence, $X = A$ is a
  solution to the instance $\langle G,h \rangle$ of
  \Pbalancedvertexseparator, and $\langle G,h \rangle$ is a
  yes-instance.
\end{proof}

\section{Mixed Strategies}
\label{sec:mixed}

In this section, we consider mixed strategies for both players.
First, consider for a moment a variant, where we are given a set of pure strategies for the defender and a set of pure strategies for the attacker, and we are asked to compute the value of the game when players are restricted to use only these strategies.
The assignment of weights to the strategies in a Nash equilibrium for such a variant can be computed in polynomial time by linear programming.
As there are exponentially many pure strategies for both players, this gives us an \EXPTIME algorithm for \Pmixdefmixatt.

A natural heuristic to solve \Pmixdefmixatt is to start with a small set of pure strategies for both players, compute a Nash equilibrium for this restricted game, and then iteratively add pure strategies that are best responses to the current mixed strategy of the opponent.
This heuristic is known as \defterm{column generation} and is a common approach to solve large linear programs.
It is not guaranteed to terminate in polynomial time, but it can be efficient in practice when there are only a few pure strategies that are relevant for the optimal mixed strategy.
In our case, even for the simple example of a clique~$K_{2n}$ and
$k=\ell=n$, it is easy to see that the optimal mixed strategies for
both players must use all possible $n$-element subsets of the $2n$
vertices.
Thus column generation would require an exponential number of
iterations to find the optimal mixed strategies.

We establish another difficulty within the column generation approach, by showing that given a mixed strategy for one player, it is \NP-hard to compute a best response for the other player.
We formalize these problems below, starting with the mixed defense.
\defdecproblem
{\Ppureattmixdef}
{A graph $G$, a mixed defense $\mathcal{D}$, a positive
  integer~$\ell$, and a rational number $\alpha$.}
{Is there an $\ell$-attack~$A$ such that $\payoffatt(G,\mathcal{D},A) \ge \alpha$?}
By \cref{thm:pureattfixdef}, we get that \Ppureattmixdef is \NP-complete even in the special case where the defense contains only one strategy with probability~$1$ and thus obtain the following corollary.
\begin{corollary}
  \label{thm:pureattmixdef}
  \Ppureattmixdef is \NP-complete.
\end{corollary}

To get a similar result for column generation for the defender, we need to be more careful.

\defdecproblem
{\Ppuredefmixatt}
{A graph $G$, a mixed attack $\mathcal{A}$, positive integer $k$ and rational number $\delta$.}
{Is there a $k$-defense~$D$ such that $\payoff(G,D,\mathcal{A}) \ge \delta$?}

Recall that the optimal response of the defender to a fixed attack can be computed in polynomial time by \cref{obs:optdef}.
Nevertheless, we show that, given a mixed attack with multiple strategies, it is \NP-hard to compute an optimal response for the defender.

\begin{restatable}[\restateref{thm:puredefmixatt}]{theorem}{TheoremPureDefMixAtt}
  \label{thm:puredefmixatt}
  \hspace*{-.8ex}\Ppuredefmixatt is \NP-complete.
\end{restatable}

\begin{proof}[Proof sketch]
  Membership in \NP is clear: Given a graph $G$, a mixed attack
  $\mathcal{A}$, a $k$-defense~$D$, and a rational number~$\delta$, one can
  verify in polynomial time whether $\payoff(G,D,\mathcal{A}) \ge \delta$.
  
  To show \NP-hardness, we reduce from the following
  \NP-hard problem~\cite{Karp72}.%
  \defdecproblem{\Psetcover}
  {A set $X$, a family $\mathcal{S}$ of subsets of $X$, and a positive integer $h$.}
  {Is there an $\mathcal{S}' \subseteq \mathcal{S}$ with $\norm{\mathcal{S}'} = h$ such that $\bigcup_{S \in \mathcal{S}'} S = X$?}

  Given an instance $\langle X,\mathcal{S},h \rangle$ of
  \Psetcover with $\mathcal{S} = \{S_1,\dots,S_m\}$, we
  construct an instance $\langle G,\mathcal{A},k,\delta \rangle$ of
  \Ppuredefmixatt.
  Let $G=K_m$, and let $V(G)=\{1,2,\dots,m\}$.  For each $x \in X$, we
  perform the attack $A_x = \{ i : x \notin S_i \}$ with probability
  $1/|X|$.  As we can assume that $\mathcal{S}$ covers~$X$,
  for every $x \in X$, we have that $|A_x| < m$.
  We let $\mathcal{A}$ be the multiset $\{A_x: x\in X\}$,
  we let $k = h$, and we let $\delta = \sum_{x \in X} (m-|A_x|)/|X|$. 
\end{proof}

\section{Restricted Graph Classes}
\label{sec:restricted}

We know from \cref{thm:pureattfixdef} that \Ppureattfixdef is \NP-hard
in general.  In the following, we show that it is polynomial-time
solvable on interval graphs and on graphs of bounded treewidth.  Hence
\Ppuredefpureatt is in \NP for these graph classes.  We start with
interval graphs.
It suffices to consider connected graphs.

\begin{restatable}[\restateref{lem:component-attack-knapsack}]{lemma}{LemmaComponentAttackKnapsack}
  \label{lem:component-attack-knapsack}
  Given a graph~$G$, a defense~$D$, an integer~$\ell$, and for each
  connected component~$C$ of~$G$ and~$0\leq\ell'\leq \ell$ an optimal
  $\ell'$-attack on~$C$ under defense~$D$, we can find an optimal
  $\ell$-attack on~$G$ under defense~$D$ in $\Oh{c\cdot \ell^2}$ time,
  where $c$ is the number of connected~components~of~$G$.
\end{restatable}

Using a sweepline algorithm, we obtain the following result for interval graphs.

\begin{restatable}[\restateref{thm:interval-pureattfixdef}]{theorem}{TheoremIntervalPureAttFixDef}
  \label{thm:interval-pureattfixdef}
  \Ppureattfixdef is polynomial-time solvable on interval graphs.
\end{restatable}

Now we turn to graphs of bounded treewidth. Treewidth is a
fundamental graph parameter and expresses the tree-similarity of a
graph.  For a formal definition, see for example~\cite{CyganFKLMPPS15ParameterizedAlgorithms}.

\begin{restatable}[\restateref{thm:tw-pureattfixdef}]{theorem}{TheoremTwPureAttFixDef}
  \label{thm:tw-pureattfixdef}
  Given a graph $G$, a tree decomposition $(T,\mathcal{B})$ of~$G$, a
  defense, and an attack size~$\ell$, \Ppureattfixdef can be solved
  in $O^*(3^\tw)$ time, where \tw is the width of $(T,\mathcal{B})$ and~$O^*$ suppresses factors polynomial in the size of the input.
\end{restatable}
\begin{proof}[Proof sketch]
  We use a dynamic program on a nice tree decomposition to obtain an optimal
  $\ell$-attack~$A$ against a defense~$D$.

  The dynamic program considers consistent labelings, which we define as follows. 
  A \defterm{labeling} of a set of vertices is a partition~$(S_\mathrm{a},S_\mathrm{d},S_\mathrm{s})$ of the vertices into attacked, disabled, and surviving vertices. 
  We say that a labeling is \defterm{consistent} (w.r.t.~$D$) if it holds that
  \begin{enumerate}[label=(\Roman*)]
    \item There is no edge between a surviving and a disabled vertex.
    \item No disabled vertex is a controller.
  \end{enumerate}

  Given an attack~$A$, we say that~$A$ induces the labeling $(S_\mathrm{a},S_\mathrm{d},S_\mathrm{s})$ where~$S_\mathrm{a}=A$, $S_\mathrm{d}$ is the set of disabled vertices under attack~$A$ and defense~$D$, and~$S_\mathrm{s}$ is the set of surviving vertices under attack~$A$ and defense~$D$.

  \begin{claim}
    For every attack~$A$, the labeling~$(S_\mathrm{a},S_\mathrm{d},S_\mathrm{s})$ induced by~$A$ is consistent. 
    Furthermore, for every consistent labeling~$(S_\mathrm{a}',S_\mathrm{d}',S_\mathrm{s}')$ with~$S_\mathrm{a}'=A$, it holds that~$S_\mathrm{d}'\subseteq S_\mathrm{d}$.
  \end{claim}
  
  For each node~$t$ of the tree decomposition, let $V_t$ be the union
  of the bags corresponding to the nodes in the subtree of~$T$ rooted
  at~$t$.  Let $G_t=G[V_t]$.  We  maintain a
  table~$P_t[\ell',S_\mathrm{a},S_\mathrm{d},S_\mathrm{s}]$, where~$S_\mathrm{a} \, \dot\cup \, S_\mathrm{d} \, \dot\cup \, S_\mathrm{s}=B_t$.  The value
  of~$P_t[\ell',S_\mathrm{a},S_\mathrm{d},S_\mathrm{s}]$ is the maximum over~$\norm{S_\mathrm{a}'}+\norm{S_\mathrm{d}'}$ for all consistent labelings~$(S_\mathrm{a}',S_\mathrm{d}',S_\mathrm{s}')$ of~$V_t$ which coincide with~$(S_\mathrm{a},S_\mathrm{d},S_\mathrm{s})$ on~$B_t$ and where~$\norm{S_\mathrm{a}'}\leq\ell'$.

  Thus, for the root~$r$ of the tree-decomposition, the entry~$P_r[\ell,\emptyset,\emptyset,\emptyset]$ contains the maximum number of disabled vertices in~$G$ by an~$\ell$-attack under defense~$D$, as~$B_r=\emptyset$ and the labeling induced by an attack has at least as many disabled vertices as any other consistent labeling with the same attacked vertices.
\end{proof}

\section{Open Problems}

From our last two theorems, it follows that \Ppuredefpureatt is in \NP for interval graphs and graphs of bounded treewidth.
While it is tempting to use similar dynamic programs to solve the more general case of \Ppuredefpureatt in polynomial time, by additionally deciding for each vertex whether it is defended, there is a fundamental issue with this approach: In \Ppuredefpureatt, the defender has to determine its whole defense before the attacker attacks. However, a dynamic program that iteratively proceeds along a graph allows the defender to adapt its defense to the attack decisions made in the previous steps, something not possible in the actual problem. 
Still, we conjecture that \Ppuredefpureatt is in \POL for interval graphs and graphs of bounded treewidth.

Another natural question is whether any of the problems that are hard
on general graphs admit polynomial-time solutions on restricted graph
classes, e.g., planar graphs. 

The hardness of \Pmixdefmixatt remains open.

\begin{credits}
\subsubsection{\ackname}
Grzegorz Gutowski was partially supported by the National Science Centre,
Poland (grant no.~2023/49/B/ST6/01738).

\lncs{%
  \subsubsection{\discintname}
  The authors have no competing interests to declare that are relevant to the content of this article.%
}
\end{credits}

\bibliographystyle{splncs04}
\bibliography{paper}

\lncs{\end{document}}

\clearpage

\appendix

\section{A Quick Introduction to the Complexity Class \SPTWO}
\label{app:sptwo}

We refer to the textbook by Arora and Barak~\cite[Chapter~5]{AroraB09}
for an introduction of the polynomial hierarchy.  Schaefer and
Umans~\cite{SchaeferU02_1,SchaeferU02_2,SchaeferU08} give an extensive
list of complete problems for different classes in the polynomial
hierarchy.  For a very brief introduction, \SPTWO is defined as
$\NP^\NP$~-- a class of languages decidable in polynomial time by
nondeterministic Turing machines with access to an \NP-oracle, which
can decide any language in \NP in a single step of execution.  The
canonical complete problem for \SPTWO is the following.
\defdecproblem{\PSPTWOSAT}{Boolean formula $\varphi(x_1,\ldots,x_a,y_1,\ldots,y_b)$ with variables in two disjoint sets $\set{x_1,\ldots,x_a}$ and $\set{y_1,\ldots,y_b}$}{Is the following Boolean formula true?\\
  &$\exists{x_1,x_2,\ldots,x_a}:\quad
  \forall{y_1,y_2,\ldots,y_b}:\quad
  \varphi(x_1,\ldots,x_a,y_1,\ldots,y_b)$ }
Stockmeyer~\cite{Stockmeyer76} and Wrathall~\cite{Wrathall76}
independently proved that the class \SPTWO is exactly the class of
languages reducible to \PSPTWOSAT via polynomial-time many-to-one
reductions.

\section{Full Proofs}
\label{app:observation}

\obsOptDef*
\label{obs:optdef*}

\begin{proof}
  Clearly, the optimal defense~$D$ defends the $k$ largest connected components of~$G-A$, by placing one controller in each of these components. 
  Thus, to obtain the optimal payoff, it suffices to compute the sizes of the connected components of~$G-A$. 
  Compute~$G'=G-A$ by removing the vertices in~$A$ from~$G$.
  Traverse~$G'$ (e.g., by breadth-first search) to compute a list (that
  is, a multiset)~$C$ of the sizes of the connected components
  of~$G'$.
  If $|C| \le k$, let $C'=C$.  If $|C| > k$, find (in
  $O(|C|)$ time) the $k$-th largest number~$c^\star$ in~$C$.  In this
  case, let~$C'$ be the sublist of $\{c \in C, c \ge c^\star\}$ that
  contains only as many copies of~$c^\star$ as needed so that
  $|C'|=k$.  Clearly, $\payoff(G,D,A)= \sum_{c \in C'} c$.
  It is easy to see that the algorithm can be implemented to run in time
  linear in the size of~$G$.
\end{proof}

\ClaimNodeVertexSurvival*
\label{claim:node-vertex-survival*}
\begin{claimproof}
  Since all node vertices form a clique in~$H$, if one node vertex survives, all non-attacked node vertices survive.

  Let~$C$ be the connected component of~$H-A$ that contains the non-attacked node vertices. 
  The vertices not in~$C$ are the attacked vertices, as well as the copy vertices of the attacked node vertices, and the edge vertices where both node vertices corresponding to the endpoints of the edge are attacked.
  Since~$t$ vertices are attacked, this results in at most~$\binom{t}{2}+2t$ vertices outside~$C$.
  Since~$s>\binom{t}{2}+2t$, there is at least one defended vertex in~$C$, resulting in all non-attacked node vertices surviving.
\end{claimproof}

\TheoremPureDefMixAtt*
\label{thm:puredefmixatt*}
\begin{proof}
  Membership in \NP is clear: Given a graph $G$, a mixed attack
  $\mathcal{A}$, a $k$-defense~$D$, and a rational number~$\delta$, one can
  verify in polynomial time whether
  $\payoff(G,D,\mathcal{A}) \ge \delta$.
  
  To show \NP-hardness, we reduce from the following classic \NP-complete problem~\cite{Karp72}.
  \defdecproblem{\Psetcover}
  {A set $X$, a family $\mathcal{S}$ of subsets of $X$, and a positive integer $h$.}
  {Is there an $\mathcal{S}' \subseteq \mathcal{S}$ with $\norm{\mathcal{S}'} = h$ such that $\bigcup_{S \in \mathcal{S}'} S = X$?}

  Given an instance $\langle X,\mathcal{S},h \rangle$ of
  \Psetcover with $\mathcal{S} = \{S_1,\dots,S_m\}$, we
  construct an instance $\langle G,\mathcal{A},k,\delta \rangle$ of
  \Ppuredefmixatt.
  Let $G=K_m$, and let $V(G)=\{1,2,\dots,m\}$.  For each $x \in X$, we
  perform the attack $A_x = \{ i : x \notin S_i \}$ with probability
  $1/|X|$.  As we can assume that $\mathcal{S}$ covers~$X$,
  for every $x \in X$, we have that $|A_x| < m$.
  We let $\mathcal{A}$ be the multiset $\{A_x: x\in X\}$, we let
  $k = h$, and we let $\delta = \sum_{x \in X} (m-|A_x|)/|X|$.

  If $\langle X,\mathcal{S},h \rangle$ is a yes-instance of
  \Psetcover, then there is a set cover
  $\mathcal{S}' \subseteq \mathcal{S}$ of size~$h$.  Let
  $D=\{i : S_i \in \mathcal{S}'\}$ and note that
  $|D|=|\mathcal{S}'|=h=k$.  Then, for every $x \in X$, the
  vertices in $[m] \setminus A_x = \{i: x \in S_i\}$ survive
  attack~$A_x$.  Hence, $\payoff(G,D,\mathcal{A}) = \delta$, and
  $\langle G,\mathcal{A},k,\delta \rangle$ is a yes-instance of \Ppuredefmixatt.

  On the other hand, if $\langle G,\mathcal{A},k,\delta \rangle$ is a
  yes-instance of \Ppuredefmixatt, then there exists a
  $k$-defense~$D$ with $\payoff(G,D,\mathcal{A}) \geq \delta$.  Clearly,
  for every~$x \in X$, $\payoff(G,D,A_x) \le m-|A_x|$.  Hence,
  $\payoff(G,D,\mathcal{A}) \geq \delta$ implies that, for
  every~$x \in X$, we actually have $\payoff(G,D,A_x) = m-|A_x|$.
  This is only possible if $D\cap ([m] \setminus A_x) \neq \emptyset$.  In
  particular, there is an~$i$ in~$D$ such that $x \in S_i$.  In other
  words, $\{ S_i : i \in D\}$ is a set cover of size~$h=k$, and
  $\langle X,\mathcal{S},h \rangle$ is a yes-instance of
  \Psetcover.
\end{proof}

\LemmaComponentAttackKnapsack*
\label{lem:component-attack-knapsack*}
\begin{proof}
  Let~$\set{C_1,\dots,C_c}$ be the connected components of~$G$ and,
  for~$i\in[c]$ and~$0\leq\ell'\leq\ell$, let~$A_{i,\ell'}$ be an
  optimal $\ell'$-attack on component~$C_i$.
  Let~$P_{i,\ell'}=\payoffatt(G[C_i],D\cap C_i,A_{i,\ell'})$. 
  Let~$G_i=G[C_1\cup\dots\cup C_i]$.
  For~$i\in[c]$ and~$0\leq\ell'\leq \ell$, let~$T[i,\ell']$ be the payoff of the optimal $\ell'$-attack on~$G_i$.
  Clearly,~$T[1,\ell']=P_{1,\ell'}$.
  For~$i>1$, it holds that $T[i,\ell']=\max\set{T[i-1,\ell'-\ell'']+P_{i,\ell''}\mid 0\leq\ell''\leq\ell'}$.
  The optimal payoff is given by~$T[c,\ell]$.  Using dynamic
  programming, each entry of~$T$ can be computed in~$\Oh{\ell}$ time.
  Since~$T$ has~$\Oh{c\cdot \ell}$ entries, the dynamic program runs
  in~$\Oh{c\cdot \ell^2}$ total time.  By backtracking through the
  table, we can construct the optimal attack within the same time bound.
\end{proof}

\TheoremIntervalPureAttFixDef*
\label{thm:interval-pureattfixdef*}
\begin{proof}
  Let~$G$ be an interval graph, and let~$D$ be a $k$-defense. In the following, we show how to find an optimal $\ell$-attack against~$D$ in polynomial time.

  Given~$G$, we can compute, in linear time, an interval representation of~$G$.
  Let~\I be the corresponding set of intervals.
  We use a sweepline algorithm, going through~\I from left to right.
  For~$x \in \Real$, let~$\I_x$ be the set of intervals in~\I that contain~$x$, let~$\I_{\le x}$ be the set of intervals in~\I whose left endpoint is to the left of (or equal to)~$x$, and let~$\I_{>x}$ be the set of intervals in~\I whose left endpoint is to the right of~$x$.
  Let~$G_{\leq x}$ be the subgraph of~$G$ induced by the intervals of~$\I_{\le x}$. 
  We assume without loss of generality that no two intervals have the same endpoints (this can be achieved by a small perturbation of the endpoints), and that the intervals are closed on the left but open on the right.
  
  As usual, the movement of the sweepline is discretized at certain eventpoints and accompanied by a status.
  Our eventpoints are the endpoints of the intervals with an additional special eventpoint~$x^-$ to the left of all intervals, which serves as our starting point.
  Due to \cref{lem:component-attack-knapsack}, we can assume without loss of generality that~$G$ is connected. 
  Our status is a table~$T$.
  For $x \in \Real \cup \{x^-\}$, $y,z \in \Real \cup \{\none\}$, and integer $0\leq j\leq \ell$,
  we let $T[x,y,z,j]$ be the maximum number of disabled vertices in~$G_{\leq x}$ under defense~$D$ using a $j$-attack with the following properties:
  (i)~the surviving interval in~$\I_x$ that reaches furthest to the right ends at~$y$ (or, if $y$ is $\none$, then no interval in~$\I_x$ survives) and
  (ii)~the earliest interval outside~$G_{\leq x}$ (i.e., from~$\I_{>x}$) that survives without being kept alive exclusively by a controller in~$G_{\leq x}$, starts at~$z$ (or, if~$z$ is $\none$, then no interval survives outside~$G_{\leq x}$ without being kept alive from inside~$G_{\leq x}$).
  Let $T[x,y,z,j]=-\infty$ if no $j$-attack with the above properties exists. 
  \cref{fig:interval-sweepline} shows an example.

  \begin{figure}[tb]
    \centering
    \includegraphics[page=3]{clique-node-deletion-reduction}
    \caption{An example of the restrictions resulting from the parameters~$y$ and~$z$ in the sweepline status. This entry considers only attacks where the red intervals are disabled, whereas the green intervals survive. The black interval may survive or not.}
    \label{fig:interval-sweepline}
  \end{figure}

  Let~$x^+$ be the rightmost endpoint of any interval.
  Clearly,~$T[x^+, \none, \none, \ell]$ is the maximum number of vertices that an $\ell$-attack on~$G$ can disable under defense~$D$, and is thus our solution. 

  Before we show how to compute the entries of~$T$, consider the requirements for an entry $T[x,y,z,j]$ to be feasible:
  It must hold that~$y\in\set{v.\rightend\mid v\in \mathcal{I}_x}\cup\{\none\}$ and that~$z\in\set{v.\leftend\mid v\in \mathcal{I}_{>x}}\cup\{\none\}$.
  In the following, we deal with the case that there are no attacks left, i.e., $j=0$.
  In this case, either all intervals in~$G_{\leq x}$ survive or none do, as none of the intervals in~$G_{\leq x}$ can be attacked and~$G_{\leq x}$ is connected. 
  Thus, if~$y=\none$, no vertex in~$G_{\leq x}$ survives, and therefore it must hold that~$G_{\leq x}$ contains no controller under~$D$ and that~$z=\none$ or~$\max_{v\in\mathcal{I}_x}v.\rightend<z$.
  In this case, the attack receives all vertices in~$G_{\leq x}$ as payoff, and it thus holds that~$T[x,\none,z,0]=\norm{\mathcal{I}_{\leq x}}$.
  The other feasible case is that all vertices in~$G_{\leq x}$ survive, meaning that~$y=\max_{v\in \mathcal{I}_x}v.\rightend$.
  In this case, it must hold that~$G_{\leq x}$ contains a controller under~$D$ or that~$y>z$, and it holds that~$T[x,y,z,0]=0$.
  All other cases are infeasible.

  The base case is~$T[x^-,\none,z,j]=0$ for arbitrary~$z$ and~$j$.

  To maintain the status, we first consider the case that the current eventpoint~$x$ is the left endpoint of an interval~$v$, i.e., $x=v.\leftend$; see \cref{alg:interval-left-endpoint} in Appendix~\ref{app:pseudocode} for a summary of this case.
  Let~$x'$ be the previous eventpoint. 
  If~$y=\none$, then no interval from~$\mathcal{I}_x$ survives. If~$v$ is a controller under~$D$, it must be attacked and it holds that~$T[x,\none,z,j]=T[x',\none,z,j-1]+1$. If~$v$ is not a controller, we have two cases: If~$z<v.\rightend$, then~$v$ must be attacked to ensure that no interval intersecting~$x$ survives and it holds that~$T[x,\none,z,j]=T[x',\none,z,j-1]+1$. Otherwise, if~$z>v.\rightend$ or~$z=\none$,~$v$ does not survive regardless of whether it is attacked or not, and it holds that~$T[x,\none,z,j]=T[x',\none,z,j]+1$.
  If~$y<v.\rightend$, then~$v$ cannot survive and must be attacked (as it is otherwise kept alive by the overlapping interval ending at~$y$), and it holds that~$T[x,y,z,j]=T[x',y,z,j-1]+1$.
  If~$y>v.\rightend$, then it is not specified whether~$v$ survives or not, and we can therefore choose whether to attack it. If~$v$ is a controller under~$D$, it holds that~$T[x,y,z,j]=\max\set{T[x',y,v.\leftend,j],T[x',y,z,j-1]+1}$. If~$v$ is not a controller, it holds that~$T[x,y,z,j]=\max\set{T[x',y,z,j],T[x',y,z,j-1]+1}$. 
  That leaves the case that~$y=v.\rightend$, meaning that~$v$ must survive. We have to consider all possible entries~$y'$ for the previous endpoint~$x'$, which can be any right endpoint of an interval in~$\mathcal{I}_{x'}$ that ends before~$v.\rightend$ or, in some cases, $\none$. Furthermore, we have to determine the entry~$z'$ for the earliest surviving interval outside~$G_{\leq x'}$. If~$v$ is a controller under~$D$ or if~$z<v.\rightend$, then~$v$ is kept alive without~$G_{\leq x'}$, and therefore~$z'=v.\leftend$, and~$y'$ can also be~$\none$. Otherwise, if~$v$ is not a controller and~$z>v.\rightend$ or~$z=\none$,~$v$ cannot be kept alive without~$G_{\leq x'}$, and therefore~$z'=z$. Moreover, $v$ must then be kept alive by a surviving interval of~$\mathcal{I}_{x'}$, and thus~$y'\neq\none$. Thus, it holds that $T[x,v.\rightend,z,j]=\max_{y'\in Y'}T[x',y',z',j]$ where~$z'=v.\leftend$ and~$Y'=\set{w.\rightend\mid w\in\mathcal{I}_{x'}, w.\rightend<v.\rightend}\cup\{\none\}$ if~$v$ is a controller under~$D$ or~$z<v.\rightend$, and~$z'=z$ and~$Y'=\set{w.\rightend\mid w\in\mathcal{I}_{x'}, w.\rightend<v.\rightend}$ otherwise.

  Next, we consider the case that the current eventpoint~$x$ is the right endpoint of an interval~$v$, i.e., $x=v.\rightend$.
  Let again~$x'$ be the previous eventpoint.
  Recall that we defined our intervals to be open on the right, meaning that~$v\notin \mathcal{I}_x$.
  Therefore, if~$y\neq \none$, it holds that~$T[x,y,z,j]=T[x',y,z,j]$ as this does not restrict the survival of~$v$.
  However, if~$y=\none$, we have to consider whether~$v$ survives or not.
  Thus, it holds that~$T[x,\none,z,j]=\max(T[x',v.\rightend,z,j], T[x',\none,z,j])$.

  It remains to show that this algorithm runs in time polynomial in $n=|V(G)|$.
  There are~$\Oh{n}$ eventpoints, and for each eventpoint there are~$\Oh{n}$ choices for~$y$,~$\Oh{n}$ choices for~$z$, and~$\Oh{n}$ options for~$j$, leading to a total of~$\Oh{n^4}$ entries in the status.
  Each entry can be computed in linear time.
  This leads to a total runtime of~$\Oh{n^5}$ for a connected component.
  Using \cref{lem:component-attack-knapsack}, we can handle disconnected graphs in polynomial time as well.
\end{proof}

\TheoremTwPureAttFixDef*
\label{thm:tw-pureattfixdef*}
\begin{proof}
  Before we present the dynamic program, we introduce some necessary definitions and notations. Let~$D$ be the given (fixed) defense.

  A \defterm{labeling} of a set of vertices is a partition~$(S_\mathrm{a},S_\mathrm{d},S_\mathrm{s})$ of the vertices into attacked, disabled, and surviving vertices. 
  We say that a labeling is \defterm{consistent} (w.r.t.~$D$) if it holds that
  \begin{enumerate}[label=(\Roman*),align=left]
    \item There is no edge between a surviving and a disabled vertex.
    \item No disabled vertex is a controller.
  \end{enumerate}

  Given an attack~$A$, we say that~$A$ induces the labeling $(S_\mathrm{a},S_\mathrm{d},S_\mathrm{s})$ where~$S_\mathrm{a}=A$, $S_\mathrm{d}$ is the set of disabled vertices under attack~$A$ and defense~$D$, and~$S_\mathrm{s}$ is the set of surviving vertices under attack~$A$ and defense~$D$.

  \begin{claim}
    For every attack~$A$, the labeling~$(S_\mathrm{a},S_\mathrm{d},S_\mathrm{s})$ induced by~$A$ is consistent. 
    Furthermore, for every consistent labeling~$(S_\mathrm{a}',S_\mathrm{d}',S_\mathrm{s}')$ with~$S_\mathrm{a}'=A$, it holds that~$S_\mathrm{d}'\subseteq S_\mathrm{d}$.
  \end{claim}
  \begin{claimproof}
    We first show that the labeling induced by an attack~$A$ is consistent. 
    A vertex~$v$ survives under attack~$A$ and defense~$D$, if and only if~$v$ has a path to a controller in~$G-A$. 
    Thus, if~$v$ is a controller, it survives, fulfilling (II). Furthermore, every neighbor of a surviving vertex that is not attacked has a path to a controller in~$G-A$ (via~$v$), and thus survives as well, fulfilling (I).

    Next, suppose that there is a consistent labeling~$S'=(S_\mathrm{a}',S_\mathrm{d}',S_\mathrm{s}')$ with~$S_\mathrm{a}'=A$ such that~$S_\mathrm{d}'\not\subseteq S_\mathrm{d}$. 
    Then there is a vertex~$v\in S_\mathrm{d}'$ which is not disabled under~$A$. Since~$S_\mathrm{a}'=A$, the vertex~$v$ must survive under~$A$. Thus, there is a path from~$v$ to a controller~$w$ in~$G-A$. Since~$S'$ is consistent, it must hold due to (II) that~$w\in S_\mathrm{s}'$. Furthermore, due to (I) it follows that each vertex along the path from~$v$ to~$w$ in~$G-A$ must also be in~$S_\mathrm{s}'$, and thus~$v\in S_\mathrm{s}'$, a contradiction.
  \end{claimproof}

  We can assume without loss of generality~\cite[Lemma
  7.4]{CyganFKLMPPS15ParameterizedAlgorithms} that the given tree
  decomposition~$(T,\mathcal{B})$ of~$G$ is \defterm{nice}, that is, $T$
  is a rooted tree with the following node types:
  For each \defterm{leaf node}~$t$ of~$T$, the corresponding
  bag~$B_t$ in $\mathcal{B}$ is empty, i.e., $B_t=\emptyset$.  A
  node~$t$ of~$T$ is an \defterm{introduce node} if it has exactly one
  child~$t'$ and it holds that~$B_t=B_{t'}\cup\set{v}$ for some
  vertex~$v$ of~$G$. We say that $v$ is \defterm{introduced} at~$t$.  A
  node~$t$ of~$T$ is a \defterm{forget node} if it has exactly one
  child~$t'$ and it holds that~$B_t=B_{t'}\setminus\set{v}$ for some
  vertex~$v$ of~$G$. We say that $v$ is \defterm{forgotten} at~$t$.  A
  node~$t$ of~$T$ is a \defterm{join node} if it has exactly two
  children~$t_\mathrm{L}$ and~$t_\mathrm{R}$ and it holds that~$B_t=B_{t_\mathrm{L}}=B_{t_\mathrm{R}}$.
  The root~$r$ of~$T$ has $B_r=\emptyset$.

  We now present a dynamic program for finding an optimal
  $\ell$-attack~$A$ against a defense~$D$. 
  
  For each node~$t$ of the tree decomposition, let $V_t$ be the union
  of the bags corresponding to the nodes in the subtree of~$T$ rooted
  at~$t$.  Let $G_t=G[V_t]$.  We maintain a
  table~$P_t[\ell',S_\mathrm{a},S_\mathrm{d},S_\mathrm{s}]$,
  where~$S_\mathrm{a} \, \dot\cup \, S_\mathrm{d} \, \dot\cup \,
  S_\mathrm{s}=B_t$.  The value
  of~$P_t[\ell',S_\mathrm{a},S_\mathrm{d},S_\mathrm{s}]$ is the
  maximum over~$\norm{S_\mathrm{a}'}+\norm{S_\mathrm{d}'}$ for all
  consistent labelings~$(S_\mathrm{a}',S_\mathrm{d}',S_\mathrm{s}')$
  of~$V_t$ that coincide
  with~$(S_\mathrm{a},S_\mathrm{d},S_\mathrm{s})$ on~$B_t$ and
  where~$\norm{S_\mathrm{a}'}\leq\ell'$.

  Thus, for the root~$r$ of the tree-decomposition, the entry~$P_r[\ell,\emptyset,\emptyset,\emptyset]$ gives the maximum number of disabled vertices in~$G$ by an~$\ell$-attack under defense~$D$, as~$B_r=\emptyset$ and the labeling induced by an attack has at least as many disabled vertices as any other consistent labeling with the same attacked vertices.

  The base case are the leaf nodes. For each leaf node~$t$, it holds that~$B_t=\emptyset$. Hence, for each~$\ell'\in\set{0,\dots,\ell}$, there is only one entry~$P_t[\ell',\emptyset,\emptyset,\emptyset]$, which we set to~$0$. If an entry is infeasible, we set it to~$-\infty$. For convenience, we implicitly add infeasible entries for~$\ell'=-1$.

  \paragraph{Introduce} Let~$t$ be an introduce node with child~$t'$, and let~$v$ be the vertex introduced at~$t$. We compute the entries of~$P_t$ as follows:

  If~$v$ is attacked, i.e., $v\in S_\mathrm{a}$, it holds that~$P_t[\ell',S_\mathrm{a},S_\mathrm{d},S_\mathrm{s}]=P_{t'}[\ell'-1,S_\mathrm{a}\setminus\set{v},S_\mathrm{d},S_\mathrm{s}]+1$.

  If~$v$ is disabled, i.e., $v\in S_\mathrm{d}$, it must hold that no neighbor of~$v$ is surviving, i.e., $N(v)\cap S_\mathrm{s}=\emptyset$. Further,~$v$ cannot be a controller, i.e., $v\not\in D$. 
  If the entry is feasible,
  it holds that~$P_t[\ell',S_\mathrm{a},S_\mathrm{d},S_\mathrm{s}]=P_{t'}[\ell',S_\mathrm{a},S_\mathrm{d}\setminus\set{v},S_\mathrm{s}]+1$, and~$P_t[\ell',S_\mathrm{a},S_\mathrm{d},S_\mathrm{s}]=-\infty$ otherwise.

  If~$v$ survives, i.e., $v\in S_\mathrm{s}$,~$v$ cannot have disabled neighbors, i.e., $N(v)\cap S_\mathrm{d}=\emptyset$. Thus, if the entry is feasible, it holds that~$P_t[\ell',S_\mathrm{a},S_\mathrm{d},S_\mathrm{s}]=P_{t'}[\ell',S_\mathrm{a},S_\mathrm{d},S_\mathrm{s}\setminus\set{v}]$, and~$P_t[\ell',S_\mathrm{a},S_\mathrm{d},S_\mathrm{s}]=-\infty$ otherwise.

  \paragraph{Forget} Let~$t$ be a forget node with child~$t'$, and let~$v$ be the vertex forgotten at~$t$. We compute the entries of~$P_t$ as follows:
  Since we have forgotten~$v$, it can be in any of the three states (attacked/disabled/survived). Thus, it holds that
  \begin{gather*}P_t[\ell',S_\mathrm{a},S_\mathrm{d},S_\mathrm{s}]=\\\max\set{P_{t'}[\ell',S_\mathrm{a}\cup\set{v},S_\mathrm{d},S_\mathrm{s}],P_{t'}[\ell',S_\mathrm{a},S_\mathrm{d}\cup\set{v},S_\mathrm{s}],P_{t'}[\ell',S_\mathrm{a},S_\mathrm{d},S_\mathrm{s}\cup\set{v}]}.
  \end{gather*}

  \paragraph{Join}
  Let~$t$ be a join node with children~$t_\mathrm{L}$ and~$t_\mathrm{R}$.  As we join the two subgraphs~$G_{t_\mathrm{L}}$ and~$G_{t_\mathrm{R}}$, the number of combined attacked and disabled vertices is the sum of the attacked and disabled vertices in both subgraphs, minus the vertices in~$S_\mathrm{a}$ and~$S_\mathrm{d}$, which are in both subgraphs and thus counted twice. 
  Furthermore, we have to consider all possible distributions of attacks onto the two subgraphs, i.e., we have to maximize over all~$\ell_\mathrm{L},\ell_\mathrm{R}$ such that~$\ell_\mathrm{L}+\ell_\mathrm{R}=\ell'+\norm{S_\mathrm{a}}$. Therefore, it holds that~$P_t[\ell',S_\mathrm{a},S_\mathrm{d},S_\mathrm{s}]=\max\set{P_{t_\mathrm{L}}[\ell_\mathrm{L},S_\mathrm{a},S_\mathrm{d},S_\mathrm{s}]+P_{t_\mathrm{R}}[\ell_\mathrm{R},S_\mathrm{a},S_\mathrm{d},S_\mathrm{s}]-(\norm{S_\mathrm{a}}+\norm{S_\mathrm{d}})\colon \ell_\mathrm{L}+\ell_\mathrm{R}=\ell'+\norm{S_\mathrm{a}}}$.

  \paragraph{Runtime}
  There are~$\Oh{3^{\tw}\cdot \ell}$ entries per node (as~$S_\mathrm{a},S_\mathrm{s},S_\mathrm{d}$ partition~$B_t$).  As each entry can be computed in polynomial time, and there are~$\Oh{\tw \cdot n}$ nodes in the nice tree decomposition, we obtain the desired runtime.
\end{proof}

\section{Pseudocode}
\label{app:pseudocode}

\vspace*{-4ex}

\begin{algorithm}[h]
    \caption{Pseudocode for computing $T[x,y,z,j]$ if~$x=v.\leftend$}
    \label{alg:interval-left-endpoint}
    \uIf{$j=0$}{
      \uIf{$y=\none$ \KwAnd $(\max_{v\in\mathcal{I}_{x}} v.\rightend<z$ \KwAnd $\mathcal{I}_{\leq x}\cap D=\emptyset)$}{
        \Return $\norm{\mathcal{I}_{\leq x}}$\;
      }
      \uElseIf{$y=\max_{v\in \mathcal{I}_x}v.\rightend$ \KwAnd $(\mathcal{I}_{\leq x}\cap D\neq\emptyset$ \KwOr $y>z)$}{
        \Return $0$\;
      }
      \Else{
        \Return $-\infty$\;
      }
    }
    \uElseIf{$y=\none$}{
      \uIf{$v\in D$ \KwOr $z<v.\rightend$}{
        \Return $T[x',\none,z,j-1]+1$\;
      }
      \Else{
        \Return $T[x',\none,z,j]+1$\;
      }
    }
    \uElseIf{$y<v.\rightend$}{
      \Return $T[x',y,z,j-1]+1$\;
    }
    \uElseIf{$y>v.\rightend$}{
      \uIf{$v\in D$}{
        \Return $\max\set{T[x',y,v.\leftend,j],T[x',y,z,j-1]+1}$\;
      }
      \Else{
        \Return $\max\set{T[x',y,z,j],T[x',y,z,j-1]+1}$\;
      }
    }
    \ElseIf{$y=v.\rightend$}{
      $R \gets \set{w.\rightend\mid w\in\mathcal{I}_{x'}, w.\rightend<v.\rightend}$\;
      $(z',Y') \gets\begin{cases}
        (v.\leftend,R\cup\{\none\}) & \text{if~$v\in D$ \KwOr $z<v.\rightend$}\\
        (z,R) & \text{otherwise}
      \end{cases}$\;
      \Return $\max_{y'\in Y'}T[x',y',z',j]$\;
    }
\end{algorithm}

\end{document}